\def \VersionArXiv {}

\documentclass[letterpaper, 10 pt, conference]{ieeeconf}  %

\IEEEoverridecommandlockouts                              %

\ifdefined\VersionArXiv
\else
	\makeatletter
	\AtBeginDocument{%
	\@ifpackageloaded{hyperref}
	{\def\@doi#1{\href{https://doi.org/#1}
		{\ttfamily https://doi.org/#1}\egroup}}
	{\def\@doi#1{\ttfamily https://doi.org/#1\egroup}}
	\def\doi{\bgroup\catcode`\_=12\relax\@doi}}
	\makeatother
\fi
\usepackage[ruled,vlined,linesnumbered]{algorithm2e}
	\SetKwInOut{Input}{input}
	\SetKwInOut{Output}{output}
\usepackage[utf8]{inputenc}
\usepackage[english]{babel}
\usepackage{booktabs}
\usepackage{csquotes}

 \usepackage{graphicx}
 \usepackage{amssymb,amsmath}%

\DeclareMathOperator*{\argmin}{arg\,min}

\usepackage[misc,geometry]{ifsym} %

\makeatletter
\def\orcidID#1{\smash{\href{https://orcid.org/#1}{\protect\raisebox{-1.25pt}{\protect\includegraphics[height=1em]{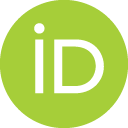}}}}}
\makeatother

\ifdefined\VersionArXiv
	\usepackage[backend=biber,backref=true,style=alphabetic,url=false,doi=true,defernumbers=true,sorting=anyt,maxnames=99]{biblatex} %
	\renewbibmacro*{doi+eprint+url}{%
		\iftoggle{bbx:doi}
			{\color{black!40}\footnotesize\printfield{doi}}
			{}%
		\newunit\newblock
		\iftoggle{bbx:eprint}
			{\usebibmacro{eprint}}
			{}%
		\newunit\newblock
		\iftoggle{bbx:url}
			{\usebibmacro{url+urldate}}
			{}%
	}

\fi
\DeclareUnicodeCharacter{0301}{\'{e}}

\ifdefined\VersionWithComments
	\usepackage{draftwatermark}
	\SetWatermarkText{draft}
	\SetWatermarkScale{2}
	\SetWatermarkColor[gray]{0.9}
\fi
\usepackage[svgnames,table]{xcolor}
\definecolor{darkblue}{rgb}{0, 0, 0.7}

\usepackage[
		pdfauthor={Jawher Jerray, Laurent Fribourg, Étienne André},%
		pdftitle={Robust optimal period control using guaranteed Euler's method},
		breaklinks  = true,
		colorlinks  = true,
	\ifdefined \VersionWithComments
		pagebackref = true,
	\fi
		citecolor   = blue!50!blue,
		linkcolor   = darkblue,
		urlcolor    = blue!50!green,
	]{hyperref}

\usepackage[capitalise,english,nameinlink]{cleveref} %
\crefname{line}{\text{line}}{\text{lines}} %

\usepackage{tikz}
\usetikzlibrary{decorations.pathmorphing}
\usetikzlibrary{arrows,calc}
\tikzset{
>=stealth',
help lines/.style={dashed, thick},
axis/.style={<->},
important line/.style={thick},
connection/.style={thick, dotted},
}

\ifdefined\VersionWithComments
	\usepackage[colorinlistoftodos,textsize=footnotesize]{todonotes}
\else
	\usepackage[disable]{todonotes}
\fi
\newcommand{\gennote}[3]{\todo[linecolor=#2,backgroundcolor=#2!25,bordercolor=#2]{#3: #1}\xspace}
\newcommand{\ea}[1]{\gennote{#1}{blue}{A}}
\newcommand{\lf}[1]{\gennote{#1}{orange}{Laurent}}

\ifdefined \VersionWithComments
	\newcommand{\todoinline}[1]{\mbox{}{\color{red}{\textbf{TODO}\ifx#1\\\else:\ \fi #1}}} %
\else
	\newcommand{\todoinline}[1]{}
\fi

\usepackage{verbatim} %

\ifdefined \VersionWithComments
	\usepackage{soul}
	
	\newcommand{\reviewDelete}[1]{{\color{red}\st{#1}}}
\else
	
	\newcommand{\reviewDelete}[1]{}
\fi

\usepackage{amsthm}

\theoremstyle{plain}
\newtheorem{lemma}{Lemma}
\newtheorem{proposition}{Proposition}
\newtheorem{theorem}{Theorem}

\theoremstyle{definition}
\newtheorem{definition}{Definition}
\newtheorem{example}{Example}

\theoremstyle{remark}
\newtheorem{remark}{Remark}

\ifdefined \VersionWithComments
 	\definecolor{colorok}{RGB}{80,80,150}
\else
	\definecolor{colorok}{RGB}{0,0,0}
\fi

\newcommand{\eg}{\textcolor{colorok}{e.\,g.,}\xspace}

\newcommand{\ie}{\textcolor{colorok}{i.\,e.,}\xspace}

\RequirePackage{amsmath}

\usepackage{mathabx}

\usepackage{color}

\title{%
Robust optimal periodic control using %
guaranteed Euler's method\ifdefined\VersionArXiv\thanks{%
	This is the author %
	version of the manuscript of the same name published in the proceedings of the 2021 American Control Conference (\href{https://acc2021.a2c2.org/}{ACC 2021}).
}\fi
}

\author{%
Jawher Jerray$^{1}$\orcidID{0000-0001-6170-7489}
\thanks{$^{1}$\href{https://lipn.univ-paris13.fr/~jerray/}{Jawher Jerray} is with Université Sorbonne Paris Nord, LIPN, CNRS, UMR 7030, F-93430, Villetaneuse, France
        \nolinkurl{jerray@lipn.univ-paris13.fr}}
\and
Laurent Fribourg$^{2}$
\thanks{$^{2}$\href{http://www.lsv.fr/~fribourg/}{Laurent Fribourg} is with the university Paris-Saclay, CNRS, ENS Paris-Saclay, LMF, F-91190 Gif-Sur-Yvette, France
        \nolinkurl{fribourg@lsv.fr}}
\and
\'Etienne Andr\'e$^{3}$\orcidID{0000-0001-8473-9555}
\thanks{$^{3}$\href{https://www.loria.science/andre/}{\'{E}tienne Andr\'{e}} is with Université de Lorraine, CNRS, Inria, LORIA, Nancy, France.
        }
}

\graphicspath{{./figures/}}

\ifdefined\VersionArXiv\fi

\begin{document}

\maketitle              %
\begin{abstract}
In this paper, we consider the application
of optimal periodic control sequences to 
switched dynamical systems. The control sequence is obtained using
a finite-horizon optimal method based on dynamic programming. We
then consider Euler approximate solutions for
the system extended with bounded perturbations.
The main result gives
a simple condition on the perturbed system for
guaranteeing the existence of a stable limit cycle of the
unperturbed system. An illustrative numerical example is
provided which demonstrates the applicability of the method.
\end{abstract}

\section{Introduction}\label{sec:intro}

When considering the optimization of
real-time processes, %
it has been shown that a
{\em periodic time-dependent} control often yields better performance than
a simple time-invariant steady-state control. This observation
has led to the creation of the field
of Optimal Periodic Control (OPC) theory in the 70's (see~\cite{Gilbert77} and  references
therein).
These periodic controls are {\em open-loop} (no feedback), so they are not  
{\em a priori} ``robust'' or ``stable''  against possible perturbations or uncertainties, and special attention must be paid 
to ensure the {\em robustness} of such controls against possible disturbances (see, \eg{}~\cite{Wang2019,DaiT12,ThuruthelFML18}).
Among recent works on new methods of robust OPC, we focus here on a line of research developed by Houska and co-workers
\cite{HouskaCDC09,SternbergACC12,SternbergIFAC12}. Their methodology 
consists in generating a ``central
optimal path'' for the case of a null perturbation, which is surrounded 
by a ``tube'', which is invariant in a robust manner
(\ie{} in the presence of a bounded perturbation $w\in{\cal W}$).
Here we consider a simplified problem compared to that of~\cite{HouskaCDC09} (cf.~\cite{NagyB03}):
we focus on the optimal open-loop control of the system {\em without} perturbation (``nominal control'') and analyze its robustness in the presence of perturbation
while~\cite{HouskaCDC09} modifies the nominal control in order to satisfy 
additional prescribed constraints on the state of the system (``robustified control'').

We keep the idea of ``tube'' used in~\cite{HouskaCDC09,SternbergACC12,SternbergIFAC12}, but we make use of recent
results related to approximate solutions
by Euler's method 
(see~\cite{LeCoentF19,CoentF19}).  
Our method makes a preliminary use of a 
dynamic programming (DP) method for generating a finite sequence
of control $\pi$ which solves a finite horizon optimal problem 
in the absence of perturbation. 
We then calculate
an approximate Euler solution of the unperturbed system denoted by $\tilde{Y}(t)$ under $\pi^*$, which corresponds to the sequence~$\pi$ applied repeatedly.
We consider the tube defined by
${\cal B}_{{\cal W}}(t)$ of the form $B(\tilde{Y}(t),\delta_{{\cal W}}(t))$\footnote{We write $B(x,d)$ to denote the ball of center $x$ and radius $d$, \ie{} the set of elements $y$ such that $\|y-x\|\leq d$, where $\|\cdot\|$ is the Euclidean norm.}
where $\tilde{Y}(t)$ is the central path, and 
$\delta_{{\cal W}}(t)$ an upperbound of the deviation due to ${\cal W}$.
The main contribution of this paper is to give a simple condition  on ${\cal B}_{{\cal W}}(t)$ which guarantees that the system is 
``stable in the presence of perturbation'' in the following sense: the unperturbed system under $\pi^*$ is guaranteed to converge towards an attractive {\em limit cycle} (LC)~${\cal L}$, and the trajectories of the 
perturbed system under $\pi^*$ are guaranteed to remain inside
${\cal B}_{{\cal W}}(t)$, which is
a ``torus'' surrounding~${\cal L}$.

In contrast with many methods of OPC using elements of the theory of
LCs,
our method does not use any 
notion of  ``Lyapunov function'' (as, \eg{} in~\cite{SternbergACC12,SternbergIFAC12}) or ``monodromy matrix'' (as, \eg{} in~\cite{HouskaCDC09}).
We also explain how to compute a rate of {\em local} contraction of the system
in order to obtain more accurate results 
than those obtained using global contraction (see, \eg{} 
\cite{AminzareS14,ManchesterCDC13}).
The simplicity of our method
is illustrated
on a classical example of bioreactor (see~\cite{HouskaCDC09}).

\subsection*{Plan of the paper}
In \cref{sec:robust}, we recall the principles of the Euler-based method,
described in \cite{LeCoentF19,CoentF19,AdrienRP17},
for finding a finite control sequence $\pi$ that solves a finite-horizon
optimal control problem.
In \cref{sec:OPC}, we give a simple condition that ensures
the robustness of the control (\cref{th:3});
the method is illustrated on the bioreactor example of~\cite{HouskaCDC09}.
We conclude in \cref{sec:conclusion}.

\section{Optimal control using Euler time integration}\label{sec:robust}
We present here the Euler-based method of optimal control synthesis
given in \cite{LeCoentF19,CoentF19,AdrienRP17}.%

\subsection{Explicit Euler time integration}\label{ss:Euler}
We consider here  a time discretization of time-step $\tau$, and 
we suppose that the control law ${\bf u}(\cdot)$ is a {\em piecewise-constant} function, which takes its
values on a {\em finite} set~$U\subset \mathbb{R}^m$, called ``modes'' (or ``control inputs'').
Given $u\in U$, let us consider the differential system controlled by~$u$: 
$$\frac{dy(t)}{dt}=f_u\big(y(t)\big)\text{.}$$
where $f_u(y(t))$ stands for $f({\bf u}(t),y(t))$ with ${\bf u}(t)=u$ for $t\in[0,\tau]$, and $y(t)\in\mathbb{R}^n$ denotes the state of the system at time $t$.
The function~$u$ is assumed to be Lipschitz continuous.
We use $Y_{y_0}^u(t)$ to denote the exact continuous solution~$y$ %
of the system at time~$t\in[0,\tau]$ under constant control $u$,
with initial condition~$y_0$.
This solution is approximated using the 
{\em explicit Euler} integration method. We use $\tilde{Y}^u_{y_0}(t)\equiv
y_0+tf_u(y_0)$ to denote Euler's approximate value of $Y^u_{y_0}(t)$ for $t\in [0,\tau]$.

Given a sequence of modes (or ``pattern'') $\pi := u_1\cdots u_k\in U^k$, we denote by
$Y_{y_0}^{\pi}(t)$
the solution of the system under  mode $u_1$ on $t\in [0,\tau)$
with initial condition~$y_0$,
extended continuously with the solution of the system under mode $u_{2}$ on $t\in[\tau,2\tau]$, and so on iteratively until mode $u_k$ on 
$t\in[(k-1)\tau,k\tau)$.
The control function ${\bf u}(\cdot)$ is thus piecewise constant  with ${\bf u}(t)=u_{i}$
for $t\in [(i-1)\tau,i\tau)$, $1\leq i\leq k$.
Likewise, we use
$\tilde{Y}_{y_0}(t)^{\pi}$ to denote Euler's approximate value of $Y_{y_0}^\pi(t)$
for $t\in [0,k\tau)$
defined by
$\tilde{Y}_{y_0}^{u_1\cdots u_i}(t)=\tilde{Y}_{y_0}^{u_1\cdots u_{i-1}}(t)+tf_{u_i}(\tilde{Y}_{y_0}^{u_1\cdots u_{i-1}}(t))$ for $t\in [0,\tau)$ and $2\leq n\leq k$.
The approximate solution $\tilde{Y}_{y_0}^{\pi}(t)$ is here a continuous piecewise linear function
on $[0,k\tau)$ starting at~$y_0$.
Note that we have supposed here that the step size $\Delta t$ used in Euler's integration method was equal to the sampling period~$\tau$ of the switching system. Actually, in order to have better  approximations,  it  is  often  convenient  to  take  a  fraction  of~$\tau$ as  for $\Delta t$ (\eg{} $\Delta t= \tau/400$). Such a splitting is called ``sub-sampling'' in numerical methods 
(see~\cite{SNR17}). Henceforth, we will suppose 
that $k\in\mathbb{N}$ is the length of the pattern $\pi$, and $T=k\tau=K\Delta t$ for some $K$ multiple of $k$, and $T>0$.

\subsection{Finite horizon and dynamic programming}\label{ss:approx}

The optimization task is to find a control pattern $\pi\in U^k$ which guarantees that all states in a given set ${\cal S}=[0,1]^n\subset \mathbb{R}^n$ \footnote{We take here ${\cal S}=[0,1]^n$ for the sake of notation simplicity, but ${\cal S}$ can be any convex subset of $\mathbb{R}^n$.} are steered 
at time $t_\mathit{end}=k\tau$  as closely as possible to an end state $y_\mathit{end}\in {\cal S}$.
Let us explain the principle of the method based on DP and Euler integration
method used in~\cite{LeCoentF19,CoentF19}.
We consider the {\em cost function}: $J_{k}:{\cal S}\times U^k\rightarrow \mathbb{R}_{\geq 0}$ 
defined by:
$$J_{k}(y,\pi)=\|Y_{y}^{\pi}(k\tau)-y_\mathit{end}\|\text{,}$$
where 
$\|\cdot\|$ denotes the Euclidean norm in $\mathbb{R}^n$\footnote{We consider here the special case where the cost function is only made of a ``terminal'' subcost. The method extends to more general cost functions. Details will be given in the extended version of this paper.}

We consider the {\em value function} ${\bf v}_k:{\cal S}\rightarrow \mathbb{R}_{\geq 0}$
defined by:
$${\bf v}_k(y) := \min_{\pi\in U^k}\big\{J_{k}(y,\pi)\big\}\equiv
\min_{\pi\in U^k} \big\{\|Y_{y}^{\pi}(k\tau)-y_\mathit{end}\|\big\}\text{.}$$ 

Given $k\in\mathbb{N}$ and $\tau\in\mathbb{R}_{>0}$, we consider the following {\em finite time horizon optimal control problem}: 
 Find for each $y\in {\cal S}$
\begin{itemize}
\item the {\em value} 
${\bf v}_k(y)$, \ie{}
$$\min_{\pi\in U^k}\big\{\|Y_{y}^{\pi}(k\tau)-y_\mathit{end}\|\big\}\text{,}$$

\item and an {\em optimal pattern}:
$$\pi_k(y) := \argmin_{\pi\in U^k} \big\{\|Y_{y}^{\pi}(k\tau)-y_\mathit{end}\|\big\}\text{.}$$
\end{itemize}
We then discretize the space ${\cal S}$ by means of a grid ${\cal X}$ such
that any point $y_0\in{\cal S}$ has an ``$\varepsilon$-representative'' $z_0\in{\cal X}$
with $\|y_0-z_0\|\leq \varepsilon$, for a given value $\varepsilon>0$.
As explained in~\cite{CoentF19}, it is easy to
construct via DP
a procedure $\mathit{PROC}_k^{\varepsilon}$ which, for any $y\in{\cal S}$,
takes its representative $z\in{\cal X}$ as input, and returns 
a pattern~$\pi_k^{\varepsilon}\in U^k$ corresponding
to an approximate optimal value of ${\bf v}_k(y)$. %

\begin{example}\label{ex:1}
We consider a biochemical process model $Y$ of continuous culture fermentation (see~\cite{HouskaCDC09} as well as \cite{AKR89,KSC93,Par00,RC08}).
Let $Y = (X, S, P) \in \mathbb{R}^3$ satisfies the differential system:   
\begin{equation*}
\begin{cases}
\overset{.}{X}(t) = -D X(t) + \mu(t) X(t)\\
\overset{.}{S}(t) = D \big( S_f(t) - S(t) \big) - \frac{\mu(t)X(t)}{Y_{x/s}}\\
\overset{.}{P}(t) = -D P + \big( \alpha\mu(t) + \beta \big) X(t)\\
\end{cases}
\label{biochemical-equations}
\end{equation*}

\noindent where $X$ denotes the biomass concentration, $S$ the substrate concentration, and $P$ the
product concentration of a continuous fermentation process.
The model is controlled by $S_f \in [S_f^\mathit{min}, S_f^\mathit{max}]$.
While the dilution rate~$D$, the biomass yield $Y_{x/s}$ , and the
product yield parameters $\alpha$ and $\beta$ are assumed to be constant
and thus independent of the actual operating condition, the
specific growth rate $\mu : \mathbb{R} \rightarrow \mathbb{R}$ of the biomass is a function
of the states:

\begin{equation*}
\mu(t) = \mu_m\frac{\left( 1-\frac{P(t)}{P_m}\right) S(t)}{K_m + S(t)+\frac{S(t)^2}{K_i}}
\label{mu-biochemical-equation}
\end{equation*}

The parameters values are as follows:
$D=0.15h^{-1}$,
$K_i=22\frac{g}{L}$,
$K_m=1.2\frac{g}{L}$,
$P_m=50\frac{g}{L}$,
$Y_{x/s}=0.4$,
$\alpha=2.2$,
$\beta=0.2h^{-1}$,
$\mu_m=0.48h^{-1}$,
$\overline{S_f}=32.9\frac{g}{L}$,
$S_f^\mathit{min}=28.7\frac{g}{L}$,
$S_f^\mathit{max}=40\frac{g}{L}$.
The goal is to maximize the average productivity presented by the cost function:
\begin{equation*}
J_k = \frac{1}{T}\int^T_0DP(t)dt
\label{cout-biochemical}
\end{equation*}

The domain ${\cal S}$ of the states $(X, S, P)$ is equal to $[4.8, 7.5] \times [11, 26] \times [17.5, 26]$.
The grid~${\cal X}$ corresponds to a discretization of
${\cal S}$, where each component is uniformly discretized into a set of $\kappa$ points.
The codomain $[28.7, 40]$ of the original continuous control function $S_f(\cdot)$ is itself discretized into a finite set~$U$, for the needs of our method. After discretization, $S_f(\cdot)$ is a 
piecewise-constant function that takes its values
in the set $U$ made of 300 values uniformly taken in $[28.7, 40]$.
The function $S_f(\cdot)$ can change its value every $\tau$ seconds.

We consider:
$\tau = 1$, $\kappa=200$, $\Delta t=1/400$, $T=t_{\mathit{end}} = 48$, $k=T/\tau=48$ and $\varepsilon=\sqrt{n}/2\kappa=\sqrt{3}/400\approx 0.004$.
For $z_0 = (X(0), S(0), P(0))=(6.52, 12.5, 22.4) \in \mathcal{X}$,
The pattern $\pi$ (repeated 4 times) output by $\mathit{PROC}_k^\varepsilon(z_0)$ is depicted on \cref{fig:biochemical-init-uncertainty-periodic-3} (bottom).
For $\pi$,  we have: $J_k = 3.642$.\footnote{By comparison, the optimal cost found in \cite{HouskaCDC09} is equal to $3.11$, but satisfies a constraint on the maximum concentration $X$, which has been ignored here.}
\end{example}

\subsection{Correctness of the method}\label{ss:error}
Given a point $y\in {\cal S}$ of $\varepsilon$-representative $z\in {\cal X}$,
and a pattern $\pi^\varepsilon_k$ returned by~$\mathit{PROC}_k^\varepsilon(z)$, we are now going to show that the distance
$\|\tilde{Y}_{z}^{\pi^\varepsilon_k}(k\tau),-y_\mathit{end}\|$ converges to~${\bf v}_k(y)$
as $\varepsilon\rightarrow 0$.
We first consider the ODE:
$\frac{dy}{dt}=f_u(y)$, and give an upper bound to the error between
the exact solution of the ODE and
its Euler approximation (see~\cite{CoentF19,SNR17}).
\begin{definition}\label{def:delta}
Let $\mu$ be a given positive constant. Let us define, for all 
$u\in U$ and $t\in [0,\tau]$,
$\delta^u_{\mu}(t)$ as follows:
$\mbox{if } \lambda_u <0:$
$$\delta^u_{\mu}(t)=\left(\mu^2 e^{\lambda_u t}+
 \frac{C_u^2}{\lambda_u^2}\left(t^2+\frac{2 t}{\lambda_u}+\frac{2}{\lambda_u^2}\left(1- e^{\lambda_u t} \right)\right)\right)^{\frac{1}{2}}$$
$\mbox{if } \lambda_u = 0:$
$$\delta^u_{\mu}(t)= \Big( \mu^2 e^{t} + C_u^2 \big(- t^2 - 2t + 2 (e^t - 1) \big) \Big)^{\frac{1}{2}}$$
$\mbox{if } \lambda_u > 0:$
$$\delta^u_{\mu}(t)=\left(\mu^2 e^{3\lambda_u t}+ 
\frac{C_u^2}{3\lambda_u^2}\left(-t^2-\frac{2t}{3\lambda_u}+\frac{2}{9\lambda_u^2}
\left(e^{3\lambda_u t}-1\right)\right)\right)^{\frac{1}{2}}$$

\noindent where $C_u$ and $\lambda_u$ are real constants specific to function $f_u$,
defined as follows:
$$C_u=\sup_{y\in {\cal S}} L_u\|f_u(y)\|\text{,}$$
\noindent where $L_u$ denotes the Lipschitz constant for $f_u$, and
$\lambda_u$ is the ``one-sided Lipschitz constant'' (or ``logarithmic Lipschitz constant'' \cite{AminzareS14}) associated to $f_u$, \ie{} the 
minimal constant such that, for all $y_1,y_2\in {\cal T}$:
\begin{equation}\label{eq:H0}
	\langle f_u(y_1)-f_u(y_2), y_1-y_2\rangle \leq \lambda_u\|y_1-y_2\|^2\text{,}
\end{equation}
\noindent where 
$\langle\cdot,\cdot\rangle$ denotes the scalar product of two vectors
of ${\cal T}$, and
${\cal T}$ is a
convex and compact overapproximation of ${\cal S}$ such that
$${\cal T}\supseteq \{Y_{y_0}^u(t)\ |\ u\in U, 0\leq t\leq \Delta t, y_0\in{\cal S}\}.$$ 
\end{definition}
The constant $\lambda_u$ can be computed using a nonlinear optimization solver (\eg{} CPLEX~\cite{cplex2009v12}) or using the Jacobian matrix of $f$ (see, \eg{}~\cite{AminzareS14}).

\begin{proposition}\label{prop:basic}\cite{SNR17}
Consider the solution $Y_{y_0}^u(t)$ of $\frac{dy}{dt}=f_u(y)$ with
initial condition~$y_0$ of $\varepsilon$-representative $z_0$
(hence such that $\|y_0-z_0\|\leq\varepsilon$), 
and the approximate
solution $\tilde{Y}_{z_0}^u(t)$ given by the explicit Euler scheme.
For all $t\in[0,\tau]$, we have:
$$\|Y_{y_0}^u(t)-\tilde{Y}_{z_0}^u(t)\|\leq \delta^u_{\varepsilon}(t)\text{.}$$
\end{proposition}
\begin{remark}
The function $\delta_\varepsilon^u(\cdot)$ is similar to the ``discrepancy function'' used in \cite{FanM15}, but it gives an upper-bound  on the distance between an exact solution and an Euler approximate solution while the discrepancy function gives an upper-bound on the distance between any two exact solutions.
\end{remark}
\cref{prop:basic} underlies the principle of our set-based method
where set of points are represented as balls centered around the 
Euler approximate values of the solutions. This illustrated in~\cref{fig:illustration}: for any initial condition $x^0$ belonging
to the ball $B(\tilde{x}^0,\delta(0))$,%
the exact solution $x^1\equiv Y_{x^0}^u(\tau)$ belongs to the ball $B(\tilde{x}^1,\delta(\tau))$ where $\tilde{x}^1 \equiv\tilde{Y}_{\tilde{x}^0}^u(\tau)$ denotes the Euler approximation of the exact solution at 
$t=\tau$, and 
$\delta^u(\tau)\equiv\delta^u_{\delta(0)}(\tau)$.
\begin{figure}[h!]
\centering
\includegraphics[scale=0.5]{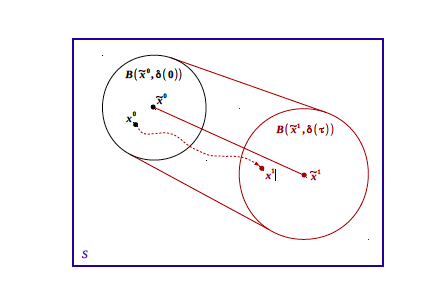}
\caption{Illustration of \cref{prop:basic}}
\label{fig:illustration}
\end{figure}
We have:

\begin{theorem}[convergence \cite{CoentF19}]\label{th:2}
Let $y\in {\cal S}$ be a point
 of $\varepsilon$-representative $z\in {\cal X}$. Let
$\pi_k^\varepsilon$ be the pattern returned by $\mathit{PROC}_k^\varepsilon(z)$,
and $\pi^\sharp := \argmin_{\pi\in U_k} \|Y^\pi_{y}(k\tau)-y_{f}\|$.
Let ${\bf v}_k(y) := \|Y_{y}^{\pi^\sharp}(k\tau)-y_\mathit{end}\|$ be the exact optimal value
of $y$. %
The approximate optimal value of~$y$,
$\|\tilde{Y}_{y}^{\pi_k^\varepsilon}(k\tau)-y_\mathit{end}\|$, converges to 
${\bf v}_k(y)$ as $\varepsilon\rightarrow 0$.
\end{theorem}

\cref{th:2} formally justifies the correctness of our method
of optimal control synthesis by saying that 
the approximate optimal values computed by our method converge to the
exact optimal values when the mesh size tends to~$0$.

\subsection{Extension to systems with perturbation}
Let us now show how the method extends to 
systems with ``bounded perturbations'', and assess its robustness.
A differential system with ``bounded perturbations'' is of the form 
$$\frac{dy(t)}{dt}=f_u\big(y(t),w(t)\big)\text{,}$$
with $u\in U$, $t\in [0,\tau]$,
states $y(t)\in\mathbb{R}^n$, and perturbations $w(t)\in{\cal W}\subset \mathbb{R}^d$ (${\cal W}$ is compact, \ie{} closed and bounded).
See, \eg{}~\cite{SchurmannA17b}.\ea{?}
Any possible perturbation trajectory is thus bounded in ${\cal W}$, and
there exists
$\omega\in\mathbb{R}_{\geq 0}$ such that $\forall t\in[0,\tau]$, $\|w(t)\|\leq \omega$.
Given a perturbation $w\in{\cal W}$, we use $Y_{y_0,w}^u(t)$ to denote the solution of
$\frac{dy(t)}{dt}=f_u(y(t),w(t))$ for $t\in[0,\tau]$ with $y(0)=y_0$.
We use
$Y_{y_0,{\bf 0}}^u(t)$
(resp.\ $\tilde{Y}_{y_0}^u(t)$)
to denote the solution (resp.\ the approximate Euler solution) without perturbations, \ie{} when ${\cal W}=0$.

Given a pattern $\pi=u_k\cdots u_1\in U^k$, these notations extend naturally to $t\in [0,k\tau]$ by considering the solutions
obtained by applying successive modes $u_k,\dots, u_1$ in a continuous manner.
The optimization task is now to find a control pattern $\pi\in U^k$ which guarantees that all states in ${\cal S}\subset \mathbb{R}^n$ are steered 
at time $t=k\tau$ as closely as possible to an end state $y_\mathit{end}$, {\em despite the
 perturbation set} ${\cal W}$. 
We suppose (see~\cite{AdrienRP17}) that, for all $u\in U$,
there exist constants $\lambda_u\in\mathbb{R}$ 
and $\gamma_u\in\mathbb{R}_{\geq 0}$ such that,
for all $y_1,y_2\in {\cal T}$ and $w_1,w_2\in {\cal W}$:
\begin{multline}\label{eq:H1}
	\langle f_u(y_1,w_1)-f_u(y_2,w_2), y_1-y_2\rangle
	\\
	\leq 
\lambda_u\|y_1-y_2\|^2 + \gamma_u \|y_1-y_2\|\|w_1-w_2\|
\end{multline}

This formula can be seen as a generalization of \cref{eq:H0} (see \cref{sec:robust}). 
Recall that $\lambda_u$ has to be computed in the absence of perturbation
(${\cal W}=0$). The additional constant $\gamma_u$ is used for taking into account
the perturbation $w$.
Given $\lambda_u$, the constant $\gamma_u$ can be computed itself
using a nonlinear optimization solver (\eg{} CPLEX~\cite{cplex2009v12}). Instead of computing them globally
for~${\cal T}$, it is advantageous to compute $\lambda_u$ and $\gamma_u$ {\em locally} depending on the subregion of~${\cal T}$ occupied by the system state during a considered interval of time $\Delta t$.
Note that the notion of contraction (often used in the literature~\cite{ManchesterCDC13,AminzareS14}) corresponds to the case where $\lambda_u$ is {\em negative} on the whole space set of interest.
(Here, $\lambda_u$ can be positive, at least locally, see \cref{rk:local}.)

We now give a version of \cref{prop:basic} with bounded perturbation $w(\cdot)\in{\cal W}$, originally proved in~\cite{AdrienRP17}.
\begin{proposition}[\cite{AdrienRP17}]\label{prop:1bis}
Consider a sampled switched system with bounded perturbation 
of the form $\{\frac{dy(t)}{dt}=f_u(y(t),w(t))\}_{u\in U}$
satisfying \cref{eq:H1}.

Consider a point~$y_0\in {\cal S}$ of $\varepsilon$-representative $z_0\in {\cal X}$.
We have,
for all  $u\in U$, $t\in[0,\tau]$ and $w(t)$ with $\|w(t)\|\leq \omega$: %

$$\|Y^{u}_{y_0, w}(t)-\tilde{Y}^u_{z_0}(t)\|\leq \delta^u_{\varepsilon,{\cal W}}(t)$$
with %
\begin{itemize}
\item if $\lambda_{u} <0$,
\begin{multline}
 \delta^u_{\varepsilon,{\cal W}}(t) = 
 \left( \frac{C_{u}^2}{-\lambda_{u}^4} \left( - \lambda_{u}^2 t^2 - 2 \lambda_{u} t + 2 e^{\lambda_{u} t} - 2 \right) \right.   \\
 + \left. \frac{1}{\lambda_{u}^2} \left( \frac{2C_{u} \gamma_{u} \omega}{-\lambda_{u}} \left( - \lambda_{u} t + e^{\lambda_{u} t} -1 \right) \right. \right.  \\ + \left. \left. \lambda_{u} \left( \frac{\gamma_{u}^2 \omega^2}{-\lambda_{u}} (e^{\lambda_{u} t } - 1) + \lambda_{u} \varepsilon^2 e^{\lambda_{u} t}  \right) \right)  \right)^{1/2}
\end{multline}
\item if $\lambda_{u} >0$,
\begin{multline}
 \delta^u_{\varepsilon,{\cal W}}(t) = \frac{1}{(3\lambda_{u})^{3/2}} \left( \frac{C_u^2}{\lambda_{u}} \left( - 9\lambda_{u}^2 t^2 - 6\lambda_{u} t + 2 e^{3\lambda_{u} t}  \right. \right.   \\
 \left. - 2 \right) + \left. 3\lambda_{u} \left( \frac{2C_u \gamma_{u} \omega}{\lambda_{u}} \left( - 3\lambda_{u} t + e^{3\lambda_{u} t} -1 \right) \right. \right.  \\
 + \left. \left. 3\lambda_{u} \left( \frac{\gamma_{u}^2 \omega^2}{\lambda_{u}} ( e^{3\lambda_{u} t } - 1) + 3\lambda_{u} \varepsilon^2 e^{3\lambda_{u} t}  \right) \right)  \right)^{1/2}
\end{multline}
\item if $\lambda_{u} = 0$, 
\begin{multline}
 \delta^u_{\varepsilon,{\cal W}}(t)= 
 \left( {C_{u}^2} \left( -  t^2 - 2  t + 2 e^{ t} - 2 \right) \right.   \\
 + \left.  \left( {2C_{u} \gamma_{u} \omega} \left( -  t + e^{ t} -1 \right) \right. \right.  \\ + \left. \left.  \left({\gamma_{u}^2 \omega^2} ( e^{ t } - 1) +  \varepsilon^2 e^{ t}  \right) \right)  \right)^{1/2}
\end{multline}
\end{itemize}
\end{proposition}

Let ${\cal B}^{u}_{{\cal W}}(t)\equiv B(\tilde{Y}^u_{z_0}(t),\delta^u_{\varepsilon,{\cal W}}(t))$.
\cref{prop:1bis} expresses that, for $t\in[0,\tau]$, the tube ${\cal B}^u_{{\cal W}}(t)$ contains all the 
solutions $Y^u_{y_0,w}(t)$ with $\|y_0-z_0\|\leq \varepsilon$ and $w\in{\cal W}$,
and is therefore {\em robustely (positive) invariant}.
The function $\delta^u_{{\cal W}}:[0,\tau]\rightarrow {\mathbb R}^n$ 
extends continuously to $\delta^\pi_{\varepsilon,{\cal W}}:[0,k\tau]\rightarrow {\mathbb R}^n$ for a sequence $\pi$ of $k$ modes, and the 
robust invariance property now holds for $t\in [0,k\tau]$. The function extends further continuously to $\delta^{\pi^*}_{\varepsilon,{\cal W}}(\cdot)$, when considering the {\em iterated application} of 
sequence $\pi$, and robust invariance property now holds for
all $t\geq 0$. Under the iterated application of $\pi$, 
we denote by $Y_{y_0,w}^{\pi^*}(t)$ %
 the exact solution at time $t$,
of the system
with perturbation $w\in{\cal W}$ and  initial condition~$y_0$. Likewise, we denote by
$\tilde{Y}^{\pi^*}_{z_0}(t)$
(or sometimes just $\tilde{Y}(t)$)
the approximate Euler solution at time $t$, of the system without perturbation,
with initial condition $z_0$.
\begin{remark}\label{rk:local}
Let us give an algorithm to compute local values of $\lambda_u$ (to simplify, we assume that $\lambda_u$ is independent of $u$). 
Given an initial ball $B_0$ with radius $d_0 := \varepsilon$,
we calculate the local value $\lambda^1$ of $\lambda_u$ and the ``successor'' ball $B_1$
of $B_0$ at $t=\Delta t$ as follows:
\begin{enumerate}
\item Select a candidate ${\cal T}_1$ for a convex zone including $B_0$ and calculate the contraction rate $-\lambda^{1}$ on ${\cal T}_{1}$.%
\item Calculate $B_1={\cal B}_{{\cal W}}(\Delta t)$ and $B'_1={\cal B}_{{\cal W}}(2\Delta t)$ using the function $\delta_{d_0,{\cal W}}$ associated with $\lambda^{1}$.
\item Check that $B_1$ and $B'_1$ are included in ${\cal T}_{1}$.
If yes, $B_1$ is indeed the successor ball (of radius $d_1=\delta_{d_0,{\cal W}}(\Delta t)$) of $B_0$; if not, go to step 1.
\end{enumerate}
We can repeat the process by taking $B_1$ as a new initial ball,
select a candidate zone ${\cal T}_2$ of rate $-\lambda^2$, calculate a ball~$B_2$ of radius $d_2=\delta_{d_1,{\cal W}}(\Delta t)$ using~$\lambda^2$, and so on iteratively.
\end{remark}
In the following we assume that the bound $\omega$ of the perturbation ${\cal W}$ is large enough so that, for all $\varepsilon\geq 0$ and all local rate of contraction $-\lambda$: 

$(H)\ \ \ \ \ \delta_{\varepsilon,{\cal W}}(\Delta t)\geq \varepsilon e^{\lambda\Delta t}$.

\section{Application to Guaranteed Robustness}\label{sec:OPC}
We suppose that a control sequence $\pi$ has been generated by $\mathit{PROC}_k^{\varepsilon}$
for solving the finite-horizon optimal control problem for the unperturbed system ($w = 0$, $T=k\tau=K\Delta t$).
We now give a simple condition on the system {\em with} perturbation ${\cal W}$ under~$\pi^*$ which guarantees the existence of a {\em stable LC~${\cal L}$} for the unperturbed system, as well as the
{\em boundedness} of the solutions of the perturbed system.
Let us consider the tube ${\cal B}_{{\cal W}}(t)\equiv B(\tilde{Y}^{\pi^*}_{z_0}(t),\delta^{\pi^*}_{\mu,{\cal W}}(t))$  for some $\mu\geq \varepsilon$.
\begin{lemma}\label{lemma:1}
Suppose
$$(*)\ \ \ \ 	{\cal B}_{{\cal W}}\big((i+K)\Delta t\big) \subset {\cal B}_{{\cal W}}(i\Delta t),\  \text{ for some } i\geq 0.$$
Then we have:	 
\begin{enumerate}
\item\label{item:1} 
The set ${\cal I}\equiv \{y\in {\cal B}_{{\cal W}}(t)\}_{t\in [i\Delta t, (i+K)\Delta t]}$ is an invariant of the perturbed system,
i.e.: if $y_0\in {\cal I}$ then $Y^{\pi^*}_{y_0,{\cal W}}(t)\in {\cal I}$ for all
$t\geq 0$.
\item\label{item:2} $\lambda^{i+1} +\cdots + \lambda^{i+K} < 0$, 
where $-\lambda^j$  ($j=i+1,\dots,i+K$)
is the local rate of contraction\footnote{See \cref{rk:local}.}  for 
the region $\{y\in {\cal B}_{{\cal W}}(t)\}_{t\in [(j-1)\Delta t,j\Delta t]}$.
\end{enumerate}
This implies that 
the distance between two solutions of the {\em unperturbed} system starting at ${\cal I}$ decreases exponentially every  $T=K\Delta t$ time-steps, and
each solution of the unperturbed system starting at ${\cal I}$
converges to an LC ${\cal L}\subset {\cal I}$.
\end{lemma}
\begin{proof} (sketch). \cref{item:1} follows easily from $(*)$.
\cref{item:2} is proved {\em ad absurdum}: Suppose $\lambda^{i+1} +\cdots + \lambda^{i+K}  \geq 0$. It follows, using $(H)$:
$\delta_{\mu,{\cal W}}((i+K)\Delta t) \geq  
e^{(\lambda^{i+1} +\dots + \lambda^{i+K}) \Delta t}\delta_{\mu,{\cal W}}(i\Delta t)
\geq \delta_{\mu,{\cal W}}(i\Delta t)$. This implies that 
the radius of ${\cal B}_{{\cal W}}((i+K)\Delta t)$ 
is greater than or equal to
the radius of ${\cal B}_{{\cal W}}(i\Delta t)$, %
which contradicts $(*)$.
So $\lambda^{i+1} +\cdots + \lambda^{i+K} < 0$. It follows
that ${\cal I}$ is a ``contraction'' region for the unperturbed system,
and every solution  starting at $y_0\in{\cal I}$
converges to an LC 
${\cal L}\subset {\cal I}$ (cf. proof of Theorem~2 in \cite{ManchesterCDC13}).\footnote{Actually, the system may also converge to an equilibrium point, but it is convenient to consider an
equilibrium as a trivial form of LC (see \cite{ManchesterCDC13}).}
\end{proof}

From \cref{lemma:1}, it easily follows:

\begin{theorem}\label{th:3}
Let $y_0\in {\cal S}$ be a point of $\varepsilon$-representative $z_0\in{\cal X}$
(so $\|y_0-z_0\|\leq \varepsilon$).
Let $T=k\tau=K\Delta t$.
Let $\pi\in U^k$ be the optimal pattern output by $\mathit{PROC}_k^\varepsilon(z_0)$ for the unperturbed system with finite horizon $T$.
Let us consider the tube ${\cal B}_{{\cal W}}(t)\equiv B(\tilde{Y}^{\pi^*}_{z_0}(t),\delta^{\pi^*}_{\mu,{\cal W}}(t))$  for some $\mu\geq \varepsilon$.
Suppose that the following inclusion condition holds:
$$(*)\ \ \ \ \ {\cal B}_{{\cal W}}\big((i+K)\Delta t\big)\subset {\cal B}_{{\cal W}}(i\Delta t) \text{ for some } i\geq 0\text{.}$$
Then:
\begin{enumerate}
\item\label{item1}
The exact solution 
$Y_{y_0,{\bf 0}}^{\pi^*}(t)$
of the unperturbed system under control~$\pi^*$
 converges to an LC ${\cal L}\subset {\cal I}$ when $t\rightarrow\infty$.
 
\item\label{item2}
For all $w\in{\cal W}$, the exact solution $Y_{y_0,w}^{\pi^*}(t)$ of
the perturbed system always remains inside 
${\cal I}$ for $t\geq i\Delta t$.
\end{enumerate}
This reflects the robustness of the perturbed system under~$\pi^*$.

\end{theorem}

\begin{remark}
In the OPC literature, it is classical to formulate the optimization
problem with an
explicit periodicity constraint of the form
$Y(T)=Y(0)$ (see, \eg{}~\cite{Gilbert77,HouskaCDC09}).
This is not needed here. Actually, at the end of the first period
$t=T$, $Y(t)$ is in general very different from $Y(0)$ with our method
(see \cref{ex:3}).

\end{remark}

\subsection*{Implementation}%
The implementation has been done in Python
and corresponds to a program of
around 500~lines.
The source code is available at
	\href{https://lipn.univ-paris13.fr/~jerray/robust/}{\nolinkurl{lipn.univ-paris13.fr/~jerray/robust/}}.
In the experiments below, the program runs on a 2.80 GHz Intel Core i7-4810MQ CPU with 8\,GiB of memory.

\begin{example}\label{ex:3}
Let us consider the system of \cref{ex:1} and the initial point $z_0\equiv(X(0),S(0),P(0))=(6.52,12.50,22.40)$ 
Let $\pi$ be the control sequence $\pi$ found 
by $\mathit{PROC}(z_0)$ for the process 
without perturbation for $\tau=1, \Delta t=1/400$  and $T= 48$ (\ie{} $k=48$,
$K=19200$). 
Here, $\lambda$ is independent of the value of the
mode $S_f$. The values of $\lambda$ and $\gamma$ are computed locally
and vary from  $+4.0$ to $-0.1$, and from $0.06$ to $0.12$ respectively.
\lf{LF: CPU time for computing {\em all} the points $z$ of ${\cal X}$???}
Let us apply the control sequence $\pi$  
repeatedly to the process {\em with} perturbation: we suppose here that
the perturbation is additive and $\|w\|\leq \omega=0.005$.
\cref{fig:biochemical-init-uncertainty-periodic-3} displays the results of
the 4 first applications of $\pi$. 
In these figures, the red curves
represent the Euler  approximation $\tilde{Y}(t)$ of the undisturbed solution as a function of time $t$ in the plans $X$, $S$ and $P$.
The green curves correspond, in the $X$, $S$ and $P$ plans, to the borders
of tube ${\cal B}_{{\cal W}}(t)\equiv B(\tilde{Y}(t),\delta_{{\cal W}}(t))$  
with $\tilde{Y}(0)=z_0$
and $\delta_{{\cal W}}(0)=\mu=1$.\footnote{It is clear that, as required by \cref{th:3}, $\mu=1>\varepsilon\approx 0.004$.}
The 10 blue curves correspond to
as many random simulations of the system with perturbation, with initial values in $B(z_0,\mu)$.
It can be seen that the blue curves always remain well inside the green tube 
${\cal B}_{{\cal W}}(t)$ which overapproximates the set of solutions
of the system with perturbation.
The values of the
coordinates of the center $\tilde{Y}(t)$ and  the radius
$\delta_{{\cal W}}(t)$ of the green tube ${\cal B}_{{\cal W}}(t)$,
at $t = 0, T, 2T, 3T$, are:

$\tilde{Y}(0)=(6.52,12.5,22.4)$, $\delta_{{\cal W}}(0)=1$;

$\tilde{Y}(T)=(6.78068367, 12.61279314, 23.98459177)$,
$\delta_{{\cal W}}(T)=0.35893$;

$\tilde{Y}(2T)=(6.77663937, 12.62347387, 23.95516391)$, 
$\delta_{{\cal W}}(2T)=0.2475$;

$\tilde{Y}(3T)=(6.77670354, 12.62331389, 23.95558776)$, 
$\delta_{{\cal W}}(3T)=0.24533$.\\
We have: ${\cal B}_{{\cal W}}(3T)\subset {\cal B}_{{\cal W}}(2T)\subset {\cal B}_{{\cal W}}(T)$  (but ${\cal B}_{{\cal W}}(T)\not\subseteq {\cal B}_{{\cal W}}(0)$).
The computation takes 480 seconds of CPU time.
It follows by \cref{th:3} that the solution of the perturbed system, 
for $t\geq 2T$ passes periodically by ${\cal B}_{{\cal W}}(2T)=B(\tilde{Y}(2T),0.2475)$,
and the solution of the unperturbed system converges to an LC
contained in ${\cal I}\equiv \{(X,S,P)\in {\cal B}_{{\cal W}}(t)\}_{t\in [T, 2T]}$.
This appears clearly on 
\cref{fig:biochemical-init-uncertainty-periodic-3}, where 
simulations of the process with perturbation corresponds to the blue lines,
and the process without perturbation to the red line.

\begin{figure}[h!]
\centering
\includegraphics[scale=0.37]{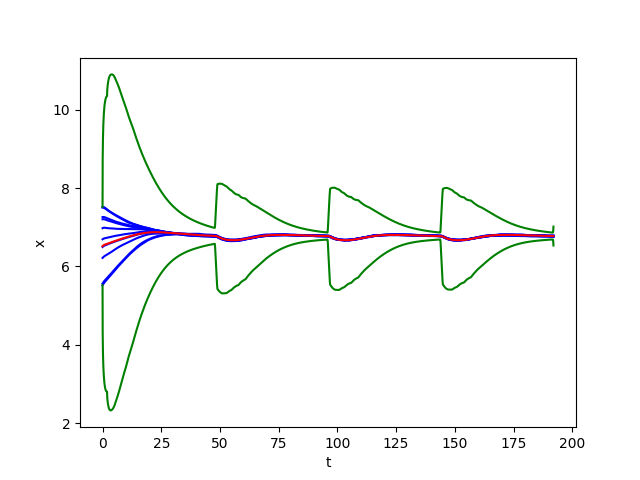}
\includegraphics[scale=0.37]{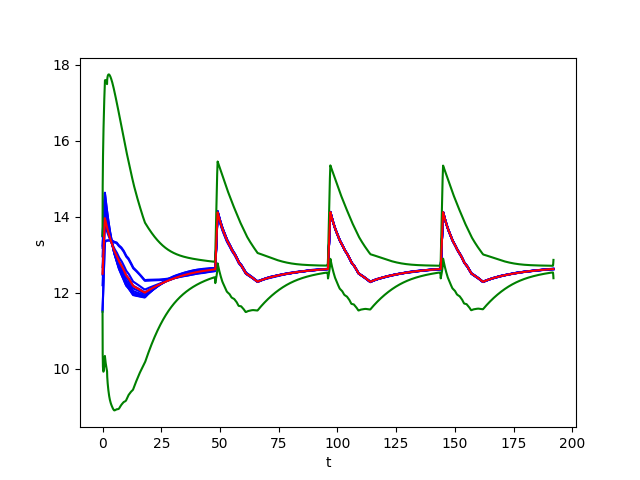}
\includegraphics[scale=0.37]{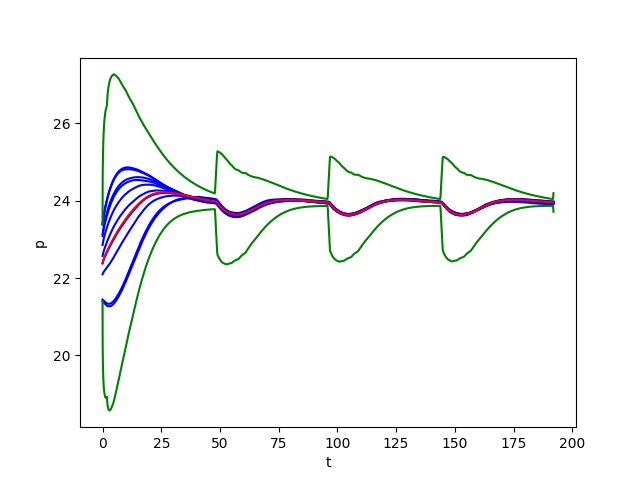}
\includegraphics[scale=0.37]{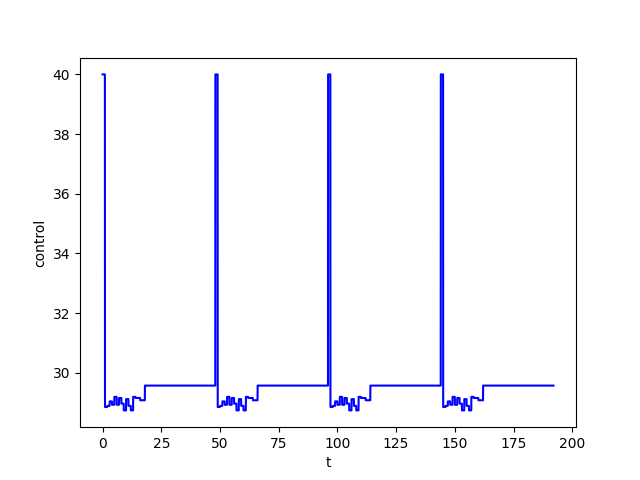}
\caption{Biochemical process with an additive perturbation $\|w\|\leq 0.005$ 
over 4 periods ($4T=192$) for $\Delta t=1/400$ and initial condition $(X(0), S(0), P(0)) = (6.52, 12.5, 22.4)$, with, from top to bottom, $X(t)$, $S(t)$, $P(t)$ and control $S_f(t)$.}
\label{fig:biochemical-init-uncertainty-periodic-3}
\end{figure}

\end{example}

\section{Conclusion}\label{sec:conclusion}

We have supposed here that a control sequence $\pi$ has been generated
for solving a finite-horizon optimal control problem for the system without perturbation ($w = 0$).
We have then given a simple condition which guarantees that,
under the repeated application $\pi^*$ of $\pi$,
the system with perturbation ($w\in{\cal W}$) is robust under $\pi^*$:
the unperturbed system is guaranteed to converge towards an LC~${\cal L}$, and
the system perturbed with ${\cal W}$ is guaranteed to stay inside a bounded tube around~${\cal L}$.
In contrast with many methods of OPC using elements of the theory of LCs
(\eg{} \cite{HouskaCDC09,SternbergACC12,SternbergIFAC12}),
the method does not make use of any
notion of monodromy matrix or Lyapunov function.
The method uses a simple algorithm 
to compute {\em local} rates of contraction in the framework of Euler's method (see \cref{rk:local}), which may be more accurate than the
global rates considered in the literature (see \eg{} 
\cite{AminzareS14,ManchesterCDC13}).
The simplicity of application of our method has been illustrated
on the example of a bioreactor given in~\cite{HouskaCDC09}.

As mentioned in \cref{sec:intro}, we have treated here a simplified problem 
of robustness compared to the one dealt with in 
\cite{HouskaCDC09} (cf.~\cite{NagyB03}).
We plan to improve our method in order to take into account 
the specification of state constraints during the evolution of the system.

\newcommand{\LNCS}{LNCS}

\ifdefined\VersionArXiv
	\renewcommand*{\bibfont}{\small}
	\printbibliography[title={References}]
\else
	\bibliographystyle{IEEEtran} %
	\bibliography{acc21}
\fi

\end{document}